\documentclass[letterpaper]{article} 
\usepackage{aaai2026}  
\usepackage{times}  
\usepackage{helvet} 
\usepackage{courier}
\usepackage[hyphens]{url} 
\usepackage{graphicx} 
\usepackage{natbib} 
\usepackage{caption}
\usepackage{algorithm}
\usepackage{algorithmic}

\usepackage{newfloat}
\usepackage{listings}
\DeclareCaptionStyle{ruled}{labelfont=normalfont,labelsep=colon,strut=off}
\floatstyle{ruled}
\newfloat{listing}{tb}{lst}{}
\floatname{listing}{Listing}

\usepackage[colorinlistoftodos,prependcaption,textsize=tiny]{todonotes}
\usepackage{dsfont}
\usepackage{amsmath}
\usepackage{amsthm}
\usepackage{enumitem}

\newcommand{\mathbbm}[1]{\text{\usefont{U}{bbm}{m}{n}#1}} 
 
\DeclareMathOperator*{\argmax}{argmax} 
\newtheorem{thm}{Theorem}
\newtheorem{dfn}[thm]{Definition} 

\newtheorem{prp}[thm]{Proposition}

\newtheorem{observation}[thm]{Observation}
\usepackage{subcaption}

\newcommand{\epsdompure}{\texttt{$\epsilon$-Dom-MILP} }
\newcommand{\epsdommixed}{\texttt{$\epsilon$-Dom-Mixed-MILP} }

\newcommand{\evaluation}{\texttt{Evaluation-MILP} }

\nocopyright
\title{Understanding Optimal Portfolios of Strategies\\ for Solving Two-player Zero-sum Games}
\author{
    Karolina Drabent,
    Ondřej Kubíček,
    Viliam Lisý
}

\affiliations{
    Artificial Intelligence Center, FEE \\
    Czech Technical University in Prague \\
    Prague, Czech Republic \\
    \{karolina.kamila.drabent, ondrej.kubicek, viliam.lisy\}@fel.cvut.cz
}

\usepackage{bibentry}

\begin{document}

\maketitle

\begin{abstract}
In large-scale games, approximating the opponent's strategy space with a small portfolio of representative strategies is a common and powerful technique. However, the construction of these portfolios often relies on domain-specific knowledge or heuristics with no theoretical guarantees. This paper establishes a formal foundation for portfolio-based strategy approximation. We define the problem of finding an optimal portfolio in two-player zero-sum games and prove that this optimization problem is NP-hard. We demonstrate that several intuitive heuristics—such as using the support of a Nash Equilibrium or building portfolios incrementally—can lead to highly suboptimal solutions. These negative results underscore the problem's difficulty and motivate the need for robust, empirically-validated heuristics. To this end, we introduce an analytical framework to bound portfolio quality and propose a methodology for evaluating heuristic approaches. Our evaluation of several heuristics shows that their success heavily depends on the specific game being solved. Our code is publicly available.\footnote{
\url{https://github.com/aicenter/portfolio_of_counterstrategies} }
\end{abstract}

\section{Introduction}
\label{sec:introduction}

Solving imperfect-information games like Poker or Stratego is a grand challenge in AI \citep{perolat2022mastering, pluribus_sandholm_2019}. The immense size of these games makes computing exact Nash equilibria intractable, necessitating the use of approximations, often through abstractions \citep{sandholm_abstraction_2015, kroer_unified_2018}. One successful abstraction technique involves restricting a player's strategy space to a small portfolio of strategies. This approach has been instrumental in state-of-the-art agents, such as those using Multi-Valued States (MVS) \citep{brown_depth-limited_2018}, where an opponent's vast strategic options are modeled via a compact portfolio.

While portfolios have appeared in various research works~\citep{bard_online_nodate,rahman_minimum_2024} on competitive games, they are often hand-crafted or constructed using heuristics with little to no theoretical backing. The central questions remain unanswered: What constitutes an optimal portfolio? How difficult is it to find one? And how can we evaluate the quality of a given portfolio?

This paper provides the first thorough investigation of these questions. Our contribution is a logical progression from theory to practice:
\begin{enumerate}[itemsep=0pt, parsep=0pt]
    \item \textbf{We formalize the problem} by defining the optimal portfolio selection problem within the clean framework of normal-form games, introducing a precise metric for portfolio quality: its exploitability.
    \item \textbf{We prove the problem is NP-hard}, establishing its inherent computational difficulty.
    \item \textbf{We demonstrate that intuitive heuristics fail}, showing via counterexamples that seemingly obvious approaches (e.g., using Nash equilibrium support or incremental construction) can yield poor solutions.
    \item \textbf{We provide a path towards principled heuristics} by introducing analytical tools for bounding portfolio quality and algorithms for computing portfolio exploitability.
    \item \textbf{We empirically evaluate} several heuristics using the proposed tools.
\end{enumerate}

While we focus on normal-form games, our work lays the theoretical groundwork necessary for developing more principled portfolio construction methods for the large-scale games where they are most needed.

\section{Related Work}
\label{sec:related_work}
An important portfolio searching algorithm is the Double Oracle (DO)~\citep{mcmahan2003planning}. It generates subsets of strategies for each player by iteratively adding best responses to the opponent’s current subset of actions. \citet{lanctot_unified_2017} introduced Policy-Space Response Oracles (PSRO), a generalized version of DO. We use DO as a baseline in our experiments.

Another class of methods is based on the transformations of the policy. In Multi-Valued States (MVS)~\citep{brown_depth-limited_2018} and Pluribus~\citep{pluribus_sandholm_2019}, a portfolio is found by "bias approach", which uses hard-coded transformations (e.g. multiplication of action probabilities) on the blueprint strategy (approximation of NE). \citet{kubicek-sepot-2024} propose an extension of the "bias approach". It uses gradients from the Regularized Nash Dynamics ~\citep{perolat2022mastering} algorithm to automatically find the transformations. We evaluate the performance of this approach and compare it to alternatives.

\citet{kroer_extensive-form_2014} define a Mixed-Integer Linear Program that computes an abstraction of an extensive form game that satisfies certain bounds. It could be adapted to create a portfolio, however it is not obvious how to adapt it to find optimal portfolios, as we define them in this paper.

\citet{rahman_minimum_2024} introduce a framework for a robust portfolio of strategies selection in Ad Hoc Teamwork Agents. This is defined in cooperative games. They introduce a concept of minimum coverage sets and use its approximation to create an algorithm for finding portfolios in the big space of strategies. However, certain definitions they use are not fully applicable in the context of competitive zero-sum games.

Implicit modeling uses portfolios for the opponent's exploitation. \citet{bard_online_2016} creates the portfolio by clustering strategies from the data. Unfortunately, the clustering method depends on having a dataset of strategies for the game that are often not available.

\section{Preliminaries}
\label{sec:notation_and_background}

In a two-player \textbf{normal-form game} $G=(A_1, A_2, u)$, $A_1$ and $A_2$ are sets of actions of player 1 and 2.  $A = A_1 \times A_2$ denotes all available action profiles, and $\Delta$ denotes a set of all probability distributions over a given set. For player $i \in \{1, 2\}$, an \textbf{(expected) utility function} $u_i: \Delta(A) \xrightarrow{} \mathcal{R} $, returns utility for a strategy profile $\pi=(\pi_1, \pi_2 ) \in \Delta(A)$.  $\pi(a)$ denotes the probability of players playing action profile $a\in A$. We refer to the opponent of player $i$ by $-i$. If $u_i(\pi) = -u_{-i}(\pi)$ for all $\pi \in \Delta(A)$, we say that the game is zero-sum. In this work, we will only consider \textbf{two-player zero-sum (2p0s) games}, hence, $u = u_1 = - u_2$.

A \textbf{pure strategy} $\pi_i$ of player $i$ is an action $a\in A_i$. A \textbf{mixed strategy} $\pi_i$ is a probability distribution over $A_i$. A strategy $BR_i(\pi_{-i})\in \Delta(A_i)$ is \textbf{best response} $BR_i(\pi_{-i}) \in \argmax_{\pi_i \in \Delta (A_i)} u_i(\pi_{i}, \pi_{-i})$. We call a strategy profile $\pi^*(G)$ a \textbf{Nash Equilibrium} (NE) of a game G if, for all agents $i$, $\pi^*_{i}$ is a best response to $\pi^*_{-i}$. We denote a set of all NE of a game $G$ as $NE(G)$. An $\epsilon$-Nash Equilibrium ($\epsilon$-NE), $\epsilon \geq 0$, is a strategy profile $\pi$ such that no player can gain more than $\epsilon$ by unilaterally deviating from their strategy. A \textbf{game value} $u^*=u(\pi^*)$ of 2p0s game is the utility value of any Nash Equilibrium  $\pi^*$. Following \citet{Cheng_Weaker_than_weak_dominance_2007}, an action $a_i \in A_i$ is \textbf{$\epsilon-$dominated}  ($\epsilon \geq 0$) iff there exists $\pi_i \in \Delta(A_i \setminus a_i)$ such that $u_i(\pi_i, \pi_{-i}) > u_i(a_i, \pi_{-i}) - \epsilon, \forall \pi_{-i} \in \Delta(A_{-i})$. A set of strategies $T\subset A_i $ is $\epsilon$-dominated iff there exists, for each $t \in T$, a mixed strategy $\pi^t_i \in \Delta(A_i \setminus T)$ such that $u_i(\pi^t_i, \pi_{-i}) > u_i(t, \pi_{-i}) - \epsilon, \forall \pi_{-i} \in \Delta(A_{-i})$.

A function $f_{NE}$, is the \textbf{equilibrium selecting function}, which for a game G, returns a single NE strategy profile  $f_{NE}(G) = \pi$, $\pi \in \text{NE}(G)$. We denote $[n] = \{1, \cdots , n\}$ the set of natural numbers up to $n$.
For $(\pi_i, \pi_{-i}) \in \Delta(A)$, \textbf{regret for action} $a_i\in A_i$ of player $i$ is  defined as $R_i(\pi_i, a_i, \pi_{-i}) = u_i(a_i, \pi_{-i}) - u_i(\pi_i, \pi_{-i})$.

In 2p0s game the \textbf{Maximum Entropy Correlated Equilibrium} (Maxent)~\cite{maxent-pmlr-v2-ortiz07a} is $\pi^{ME} = \argmax_{\pi^* \in NE} H(\pi)$, where $H(\pi)$ is the (Shannon) entropy of the strategy $H(\pi) = \sum_{a\in A} \pi(a)ln(\frac{1}{\pi(a)})$. 

\section{Portfolios}
\label{sec:Portfolios}
A portfolio is a subset of possible actions or strategies that a player will use instead of their actions in the original game. We now formally define a portfolio.
\begin{dfn}
    A \textit{(Mixed) Portfolio} of size $k \in \mathcal{N}^+$ of player $i$ in the normal form game $G$ is a subset of mixed strategies $P = \{ \pi_{1}, \pi_{2}, ... , \pi_{k} \}, P \subseteq \Delta (A_i)$, so that $|P|=k$.
\end{dfn}
 
Additionally, we define \textit{Pure Portfolio} $P$ as such that consist only of pure strategies, which means that portfolio $P$ of player $i$ is a subset of actions $P \subseteq A_i$. We will refer to the z-th strategy in the portfolio $P=\{\pi_1,...,\pi_{|P|}\}$ as $P(z)=\pi_z, z\in[|P|]$.

A normal form game $G$ can be restricted by a portfolio for a player $i$. In that case, that player is allowed to only choose the actions from the portfolio.
\begin{dfn}
    
    A 2p0s game $G=(A_1,A_2, u)$ restricted by portfolios $P_1,P_2$ is called a \textit{restricted normal form game} $G(P_1, P_2) = (P_1, P_2, u')$, where utility function is constructed with $u'(\pi_1, \pi_2) = \sum_{z_1=1}^{|P_1|} \sum_{z_2=1}^{|P_2|} \pi_1(z_1)\cdot \pi_2(z_2) \cdot u(P_1(z_1),P_2(z_2)), \pi_i \in \Delta(P_i)$.
    
\end{dfn}

\subsection{Strategy Optimization}
\label{subsec:strategy_optimization}

\begin{figure}[t] 
    \centering
    \includegraphics[width=8.5cm]{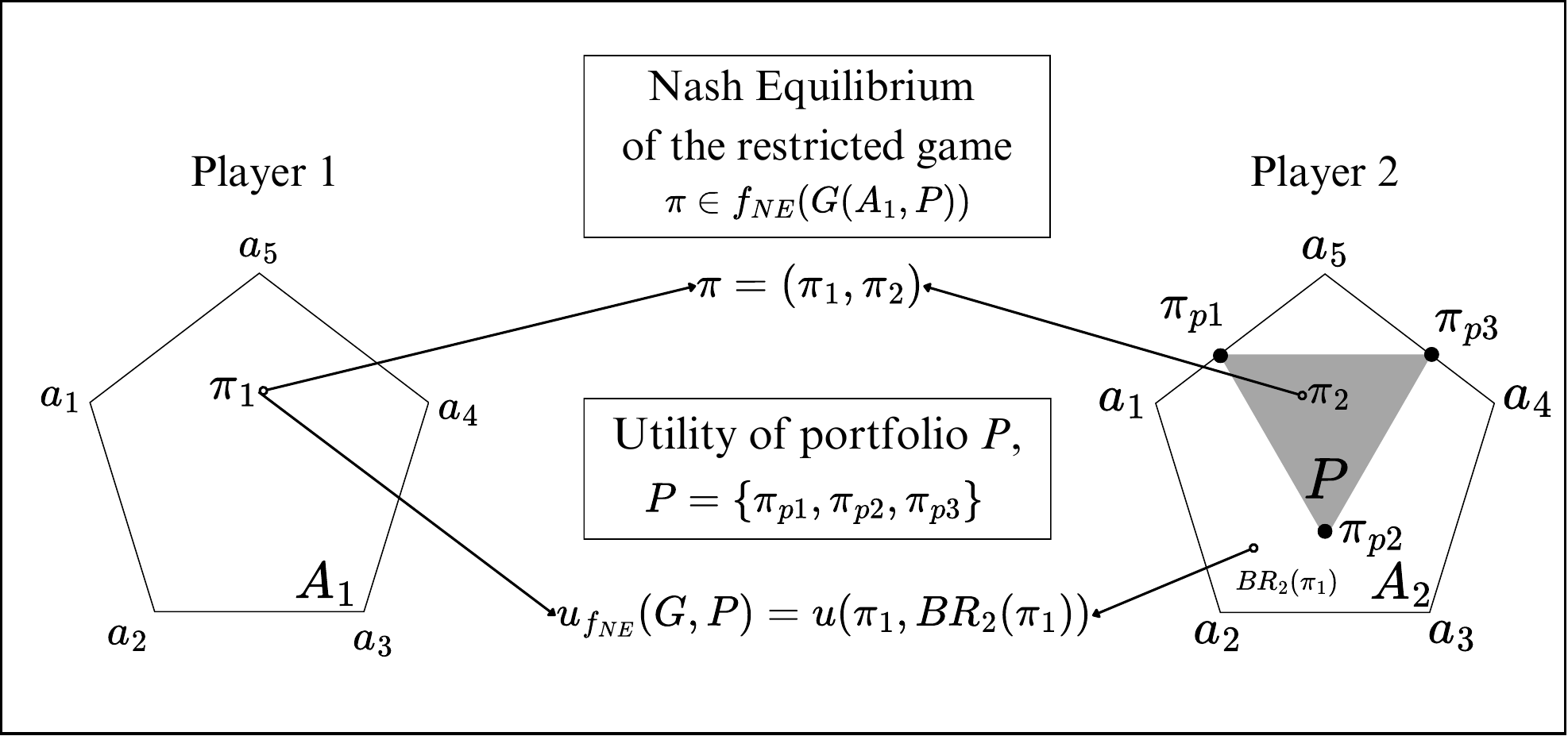}
    \caption{This diagram illustrates the strategy optimization process using a portfolio $P=\{\pi_{p1},\pi_{p2},\pi_{p3} \}$ in a game $G=(A_1, A_2,u)$. The strategy space is a multidimensional simplex mapped to two dimensions for simplicity. On the left, the simplex represents the strategy space of Player 1, $\Delta(A_1)$, while on the right, of Player 2, $\Delta(A_2)$. Within the latter, the gray simplex represents the strategy space defined by the portfolio $\Delta(P)$. The NE of the restricted game $G(A_1, P)$, $\pi, \pi \in f_{NE}(G(A_1, P))$, is computed for $\Delta(A_1)$ and $\Delta(P)$. Then, utility of portfolio P, $u_{f_{NE}}(G,P)$ is computed by fixing Player 1's strategy and finding best response of Player 2 in the full strategy space $\Delta(A_2)$.}
    \label{fig:strategy_optimization_scheme_simplex}
\end{figure}

In this paper, we focus on identifying the opponent's portfolio that enables the best performance in the original game, as in e.g. MVS~\citep{brown_depth-limited_2018}. We define this process within the normal form games and assume that strategy optimization is conducted for player 1 while player 2 is restricted by the portfolio. Initially, we are given a two-player zero-sum normal form game, $G=(A_1, A_2, u)$, and a portfolio $P$ of size $k$. Then a restricted game, $G(A_1, P)$, is constructed and one of its NE, $\pi^*_R=f_{NE}( G(A_1, P))$, is computed. Player 1 adopts this strategy in the original game $G$, while Player 2 plays best response, $BR_2(\pi^*_{R1})$. This defines portfolio utility, i.e., how much a player can lose by assuming the portfolio while facing an opponent not restricted by the portfolio. This process is illustrated in Figure~\ref{fig:strategy_optimization_scheme_simplex}.

\begin{dfn}
The \textit{utility of portfolio} $P$ in a game $G$ and for a NE selecting function $f_{NE}$ is the utility that player 1 gets when assuming the portfolio:
$u_{f_{NE}}(G,P) =  \min_{\substack{\pi'_2 \in \Delta(A_2)}}  u(\pi_{1}, \pi'_2)$, where $\pi= (\pi_1,\pi_2) = f_{NE}(G(A_1,P))  $.
\end{dfn}

In this work, we focus on three kinds of portfolio utility, based on equilibrium selection functions:
\begin{itemize}
    \item Pessimistic utility of portfolio $P$: $u_{PES}(G,P) = 
    \min_{\substack{\pi \in NE(G(A_1, P))}} \min_{\substack{\pi' \in \Delta(A_2)}} u(\pi_1, \pi'_2)$. The pessimistic equilibrium selection function finds such an equilibrium that in the original game, player 1 will have the lowest reward (to the player 2's best response). 
    \item Maxent utility of portfolio $u_{ME}(G,P)=\min_{\substack{\pi'_2 \in \Delta(A_2)}} u(\pi^{ME}_1, \pi'_2)$ where $\pi^{ME}$ is the Maximum Entropy Correlated Equilibrium.
    \item RM+ utility of portfolio is $u_{RM+}(G,P) = \min_{\substack{\pi'_2 \in \Delta(A_2)}} u(\pi_1, \pi'_2)$, where $\pi = \text{RM+}(G(A_1,P))$. RM+ (Regret Matching +) function returns the equilibrium which RM+ would return.
    \item The Optimistic Utility of Portfolio $u_{OPT}(G,P) = 
    \max_{\substack{\pi' \in NE(G(A_1,P))}} \min_{\substack{\pi_2 \in \Delta(A_2)}} u(\pi'_1, \pi_2)$. Optimistic function is the opposite of pessimistic function and returns an equilibrium in which player 1 would have the highest utility in the original game.
\end{itemize}

For any game $G$, portfolio $P$ and equilibrium selection function $f_{NE}$ the portfolio utility $u_{f_{NE}}$ is bounded by optimistic and pessimistic utilities: $u_{OPT}(G,P) \geq u_{f_{NE}}(G,P)\geq u_{PES}(G,P)$. However, computing these utilities requires knowledge beyond the restricted game. Nevertheless, we use the pessimistic case as it reflects the portfolio's performance in the worst-case scenario. We also consider $ex_{RM+}$ to be important, as in practice the NE is found by the CFR algorithm which in normal form games is equivalent to Regret Matching+~\cite{farina2023regretmatchinginstabilityfast}.

Since the objective is to find a strategy for player 1 the closest to NE strategy ($\pi =(\pi_1, \pi_2) \in\epsilon$-NE with smallest $\epsilon$), we want to know how much worse the player is for assuming the portfolio $u_{f_{NE}}(G,P)$ then playing without any restriction on actions $u(\pi), \pi \in NE(G)$. In other words, we want to know the  exploitability of using the portfolio:
\begin{dfn}
    For an equilibrium selection function $f_{NE}$ an exploitability of portfolio of player 2 $P$ in the 2p0s game $G$ is defined as exploitability of strategy $\pi, \pi \in f_{NE}(G(A_1,P)$.
    $ex_{f_{NE}}(G, P)= u(\pi^*)-u_{f_{NE}}(G, P), \pi^* \in \text{NE(G)}$.
    \label{dfn:exploitability}
\end{dfn}

\noindent
\textbf{The optimal portfolio} is one with minimal exploitability.

\section{The Hardness of the Optimal Portfolio Selection}

Having formalized the optimal portfolio selection problem, we now investigate its difficulty. We first establish that the problem is NP-hard, formally justifying the need for heuristics. We then demonstrate that several intuitive and widely used heuristic approaches are fundamentally flawed, as they can lead to arbitrarily poor portfolios. Lastly, we show that the game structure is important in this problem.

\subsection{Computational Complexity}
Our first main result establishes that finding an optimal portfolio is computationally intractable.
\begin{thm}
    Let $n$ be the number of pure strategies of Player 2 and $k<n$ the desired size of a portfolio. Deciding whether there is a portfolio of size $k$ with an exploitability lower than $\frac{1}{2n}$ is an NP hard problem. Consequently, finding the optimal portfolio of a given size is also NP hard.
\end{thm}
\textbf{Proof sketch:} The proof is by reduction from the set cover problem, using the same game construction as \cite{gilboa1989nash} use to prove that the problem of finding the minimum support Nash equilibrium is NP hard, with the roles of the players reversed. 
To align with the original proof, we assume we choose a portfolio for Player 1 to approximate a NE strategy of Player 2.

In the proof, Player 1 has to pick a set of $k$ strategies to form the set cover. This guarantees her a payoff of at least $\frac{1}{n}$. If the set is not covered, Player 2 can pick a strategy corresponding to the element that is not covered. This gives her a reward of 0, while Player 1 can trivially ensure a reward of at least $\frac{1}{2n}$ by choosing a single action. Hence, any portfolio that is not fully covering the set has a cost of at least $\frac{1}{2n}$.$\qed$

\subsection{Problems with the Intuitive Heuristics}

The NP-hardness of the problem necessitates the use of heuristics. However, are the intuitive heuristics effective? In this subsection, we demonstrate that several natural ideas for constructing portfolios are unreliable and can produce solutions with high exploitability.

\subsubsection{The Nash Equilibrium Support}

It's tempting to construct a portfolio from the strategies used in the Nash Equilibrium. The intuition is that strategies within its support are, in some sense, essential to optimal play. However, this intuition is misleading.

\begin{thm}
\label{thm:arbitrary_ne}
Using the support of an arbitrary Nash equilibrium may lead to the pessimistic exploitability of $\frac{1}{2}\Delta$ and the maximum entropy exploitability approaching $\frac{1}{2}\Delta$, where $\Delta$ is the payoff range.
\end{thm}

\begin{thm}
\label{thm:unique_ne}
Using the unique support of the column player over all Nash Equilibria as the portfolio can have an exploitability arbitrarily close to $\frac{1}{2}\Delta$.
\end{thm}

\textbf{Proof}
We will demonstrate it on the following matrix game for any small $\delta > 0$:
$$
U = \begin{bmatrix}
1 & \delta & \frac{1}{2} \\
\delta & 1 & \frac{1}{2} \\
0 & 0 & \frac{1}{2} \\
\end{bmatrix}
$$

This game has a unique equilibrium $(0,0,1)$ for player 2. When only its support is in the portfolio, in the pessimistic case, player 1 would choose to play $(0,0,1)$, which results in player 2 exploiting with action $(1,0,0)$ or $(0,1,0)$, resulting in exploitability of $-\frac{1}{2}$
$\qed$

This result demonstrates that NE support is not a reliable indicator of a good portfolio. The strategies that are crucial for maintaining an equilibrium against a rational opponent are not necessarily the same strategies that form a robust portfolio. Furthermore, not even the size of a strong portfolio can be determined by the size of the support.

\begin{observation}
\label{obs:size_portfolio_suppport}
    The size of the portfolio with 0 pessimistic exploitability is not necessarily equal to the size of the support of NE of player 2.
\end{observation}
As the proofs of Theorem~\ref{thm:arbitrary_ne} and the Observation~\ref{obs:size_portfolio_suppport} are similar to the above proof, they are provided in Appendix~\ref{app:proof_arbitrary_ne}.

\noindent
This brings us to the observation about the Double Oracle. 

\begin{observation}
    Double-oracle may find bad portfolios.
\end{observation}

Since DO ends as soon as it finds a support, it may happen that it will find the support from the counterexamples above and return the bad portfolio.

\subsubsection{Incremental Portfolio Generation} Another heuristic builds a portfolio greedily, starting with the best single strategy and gradually expanding it, hoping exploitability decreases with size. This is also not guaranteed to work.

\begin{thm}
\label{thm:incremental}
    There are games in which pure portfolio $P$ of size $k, |P|=k$ has lower exploitability than pure portfolio $P'=P \cup a, |P'| = k+1 $, for any action $a\in A_2\setminus P$.
\end{thm}
Proof can be found in the Appendix~\ref{app:proof_incremental}.
This non-monotonicity means that simple greedy or incremental algorithms are unlikely to guarantee optimality.

\subsection{Game Structure}
\label{subsec:observations}

Finally, beyond the failure of specific heuristics, some games are structured in a way that fundamentally resists approximation by small portfolios.

\begin{thm}
    \label{thm:game_max_expl}
    There are games with $n$ pure strategies for player 2, in which any pure portfolio of size $k < n$ has the maximal possible pessimistic exploitability.
\end{thm}

\begin{thm}
    \label{thm:game_pes_exploit}
    There are games with $n$ pure strategies for player 2, in which any portfolio of size $k$ has the pessimistic exploitability at least $\frac{n-k}{n k}\Delta$. Where $\Delta$ is the payoff range.
\end{thm}
The proofs for the above two theorems are in the Appendix~\ref{app:bad_games_proof}.

On the other hand, we show that high rank doesn't indicate high exploitability.
\begin{thm}
\label{thm:k_less_col_rank}

There are games in which a portfolio of a constant size, smaller than the column rank of the utility matrix, has exploitability 0.
\end{thm}

\begin{proof}
The utility matrix of such a game:
 $$
U = \begin{bmatrix}
-I & \mathbf{0}   \\
\mathbf{0} & 1 \\
\end{bmatrix}
$$
where $I$ is an identity matrix of $n \in \mathcal{N_+}$ columns. Matrix U has the rank of $n+1 = |A_2| $ and a portfolio consisting only last column has exploitability 0.
\end{proof}

Furthermore, this shows that the problem can exhibit "discontinuous" behavior. For example, a game that requires all n strategies in its portfolio (as a game with utility matrix $U=-I$) can be transformed into a trivial game (solvable with a single strategy) by the addition of one new strategy (as the game from the above proof).

This sensitivity means that even games that appear structurally similar can require vastly different portfolios. Together, these findings illustrate the profound difficulty of the optimal portfolio selection and highlight the need for the analytical tools we introduce in the following sections.

\section{Towards the principled heuristics}

The previous section painted a challenging picture: finding an optimal portfolio is computationally hard, and the most intuitive heuristics can lead to arbitrarily poor solutions. This motivates two things: (1) finding alternative analytical tools to reason about portfolio quality, and (2) recognizing the fundamental need for robust empirical evaluation of any proposed heuristic.

\subsection{Constructing Portfolios with Quality Guarantees via $\epsilon$-Dominance}

\subsubsection{$\epsilon$-Dominance and Pure Portfolios }

We propose a method for constructing pure portfolios that come with a formal guarantee on their maximum exploitability. This approach is based on the concept of $\epsilon$-dominance. 

\begin{observation}
    Pure portfolio $P\subseteq A_2$ which $\epsilon$-dominates $A_2 \setminus P$, has an upper bound on pessimistic exploitability $\epsilon$.
\end{observation}

\cite{Cheng_Weaker_than_weak_dominance_2007} show that removing the set of $\epsilon$-dominated strategies at once causes the found equilibrium to be $\epsilon$-NE  in the original game. In the context of portfolio selection, it means that if in game $G=(A_1,A_2,u)$ a set $A'_2 \subseteq A_2$ is $\epsilon$-dominated, then exploitability of portfolio $A_2 \setminus A'_2$ is bounded $ex_{PES}(G, A_2 \setminus A'_2)\leq \epsilon$. 

Therefore, we propose a Mixed-Integer Linear Program(MILP) \epsdompure that finds a portfolio $P=A_2\setminus A'_2$ of size $k$, where $A'_2\subseteq A_2, |A'_2|=|A_2|-k$ is an $\epsilon$-dominated set with minimal $\epsilon$. Following MILP finds strategies $l_j \in \Delta(A_2), j\in [|A_2|]$ that $\epsilon$-dominate all other strategies of player 2. Because there always exists pure BR, it's sufficient to enumerate through all pure actions $A_2$, as this will give the highest utility. These strategies are restricted by binary portfolio variables $b$, constraining which actions from $A_2$ can be used. Lastly, $\epsilon \in \langle 0, U_{max} - U_{min} \rangle$ is minimized.

We assume that the utility matrix $U: A_1\times A_2$ is normalized so the values are from $[-1, 1]$. That allows us to use a big enough constant $M=10$. 

\begin{align*}
\min \epsilon \\
\quad \sum_{j=1}^{|A_2|} b_j = k ,\quad \forall_{j \in [|A_2|]}\sum_{h=1}^{|A_2|} l_{jh} = 1 ,\\
\quad \forall_{j,h \in [|A_2|]} l_{jh} \leq b_h, 
\quad \forall_{\substack{ i \in [|A_1|]\\ j \in [|A_2|]}} l_{j} U_{i}^T \leq U_{ij} + \epsilon
\end{align*}

We utilize weak inequality for the $\epsilon$-dominance in the above program, as it is irrelevant when searching for the set with minimal $\epsilon$. 

With \epsdompure, finding the pure portfolio with the upper bound, the natural question that arises is whether it would be worth designing it for mixed portfolios.

\subsubsection{The Benefit of Mixed Portfolio}

For the same size of portfolio, mixed portfolios are more expressive than pure.

\begin{observation}
     For any portfolio utility $u_{f_{NE}}$, any game $G=(A_1,A_2,u)$ and portfolio size $k\leq|N|$, there exists mixed portfolio $P_{m}$, $|P_m|=k$ that's exploitability is at least as low as any pure portfolio $P_p$, $|P_p|=k$:
     $\exists_{P_m \subseteq \Delta(A_2)} \forall_{P_p \subseteq A_2} ex_{f_{NE}}(G, P_m) \leq ex_{f_{NE}}(G,P_p)$.
     \label{obs:mixed_pure}
\end{observation}

\begin{prp}
    There exist games and portfolio sizes for which the above inequality is strict.
     \label{prp:mixed_vs_pure}
\end{prp}

We provide the proofs for Observation \ref{obs:mixed_pure} and Proposition~\ref{prp:mixed_vs_pure} in Appendix~\ref{app:proof_mixed_vs_pure}. 

\subsubsection{$\epsilon$-Dominance and Mixed Portfolio}
Following on the fact that using mixed strategies in the portfolio can significantly lower the exploitability, we expand the use of the concept of $\epsilon$-dominance to mixed portfolios. We further present a theorem that allows us to assume the same upper bound for mixed portfolios. 

\begin{thm}
If a portfolio $\epsilon$-dominates all strategies of player 2 (i.e., for any pure strategy $j$ of player 2 and any strategy $x$ of player 1, a linear combination of strategies in the portfolio is at most $\epsilon$ worse against $x$ than $j$), the pessimistic exploitability of the portfolio is less than $\epsilon$.
\label{thm:esp_dom_exploitability}
\end{thm}
\textbf{Proof sketch:}
Let $\bar{x}$ be a Nash equilibrium strategy for player 1 in the game restricted by the portfolio with game value $\bar{v}$. For contradiction, assume there is a strategy $j$ of player 2 in the full game, which gains a value $u(\bar{x},j) < \bar{v}-\epsilon$. If this strategy is not included in the portfolio, it must be $\epsilon$-dominated by the portfolio. It means there is a linear combination of strategies in the portfolio $p$, such that for any $x$ (including $\bar{x}$), $u(\bar{x},j) + \epsilon \geq u(\bar{x},p) \geq \bar{v}$. The last inequality is from $\bar{x}$ being the equilibrium of the restricted game. This means that $u(\bar{x},j) \geq \bar{v} - \epsilon$, which is a contradiction.$\qed$

Based on this result, we can approximate the problem of finding the best possible portfolio of size k by searching for a set of k mixed strategies that minimizes the 
epsilon required to dominate all other strategies. We solve this using the following Mixed-Integer Linear Program, which we call \epsdommixed. It minimises $\epsilon \in \langle 0, U_{max} - U_{min} \rangle$ so that each of mixed strategies $l_z \in \Delta(A_2), z \in [k]$,  can be used to $\epsilon$-dominate any strategy $\pi_2 \in \Delta(A_2)$. For each strategy of player 2, only one of the strategies $l_z$ can be used. That is controlled by binary variables $d_{zj}$. 

\begin{align*}
\forall_{z \in [k]}\sum_{j=1}^{|A_2|} l_{zj} = 1, \quad \forall_{j \in [|A_2|]}\sum_{z=1}^{k} d_{zj} = 1,\\
\quad \forall_{\substack{z \in [k]\\ i \in [|A_1|] \\ j \in [|A_2|] }} l_{z} U_{i}^T \leq U_{ij} + \epsilon + M\cdot (1- d_{zj})\\
\min \epsilon
\end{align*}

\begin{prp}
The exploitability of the portfolio found by \epsdommixed,  $ex_{PES}(P)$, is bound by the $\epsilon$.
\end{prp}

\noindent
Proof of the proposition can be found in the Appendix~\ref{app:proof_eps_dom_mixed_milp}.

\subsection{Pessimistic Portfolio Evaluation} 
\label{sec:alg_eval}

The lack of theoretical guarantees for any tractable heuristic means we must rely on well-designed experiments to compare and validate different approaches. Consequently, we present \evaluation program for the pessimistic exploitability evaluation of (mixed) portfolio $P: k \times |A_2|$  in the 2p0s game  $G=(A_1, A_2, u)$. 

For this program value $v_R$ of the restricted game $G(A_1, P)$ needs to be computed first. Additionally, we calculate utility matrix $U_R: |A_1| \times k$ of $G(A_1,P)$ as $U_R = P U$. By $e_i \in \{0, 1\}^{|A_2|}$ we denote a vector with 1 on position $i$ and 0 otherwise. Variables $x \in [0,1]^{|A_1|}$ denote player 1's strategy and $v_o \in [min(U), max(U)]$ is the utility of portfolio. Lastly, binary variables $b \in \{0, 1\}^{|A_2|}$ signify which action is best response of player 2. We assume that the utility matrix $U: A_1\times A_2$ is normalized so the values are from $[-1, 1]$. That allows us to set a constant $M=10$.
In the following MILP, we make sure that the strategy is a distribution (1) and that there is only one best response picked (2). We also fix player 1's strategy to be NE (3) and bound best response value $v_o$ (4). 
 \[ 1) \sum_{i \in \{1,.., |A_1| \}} x_i = 1, \quad 2) \sum_{i \in \{1,.., |A_2| \}} b_i = 1,\] \[\quad 3) \;x U_R \geq v_r  \mathbbm{1}  \]
 \[ 4)\; \forall_{ i \in \{1,...,|A_2|\}} x U e_i \leq v_o + M (1 - b_i)    \]
 \begin{align*}     
     \min v_o
 \end{align*}

The exploitability of a portfolio is computed as $| v - v_{o}|$, where $v$ is the game value.

\section{Experiments}
\label{sec:experiments}

Our theoretical investigation established that finding optimal portfolios is NP-hard and that simple heuristics are unreliable. In response, we developed principled tools: \epsdompure and \epsdommixed to construct portfolios with a provable quality bound, and \evaluation to precisely measure their exploitability. In this section, we empirically validate these tools and demonstrate their practical effectiveness. Our experiments are designed to answer three key questions:
\begin{enumerate}
    \item How tight is the relationship between the theoretical $\epsilon$-dominance bound and actual portfolio exploitability?

    \item  What is the significance of finding the $\epsilon$-dominating set globally, and what is the difference in practice between \epsdommixed and \epsdompure?
    
    \item How do our methods perform against common and state-of-the-art heuristics on well-known benchmark games?
\end{enumerate}
All experiments were run on a CPU cluster. The code is publicly available at \texttt{zip supplementary materials}.

\textbf{Games domain}
In our experiments, we use two game domains: random games for robust statistical analysis and normal-form representations of well-known benchmark games (Goofspiel, Blotto, Kuhn Poker) to assess performance on games with meaningful strategic structure.

Random games~\citep{random_games_ling2019largescalelearningagent} enable easy generation of games at any size. We create them by sampling utility integer values from $[-1e07, 1e07]$ and normalizing them to the range [-1, 1].

We transformed three extensive form games to the normal form using Openspiel~\cite{lanctot2020openspielframeworkreinforcementlearning} and PyGambit~\cite{gambit2025}. Firstly, Goofspiel 3, is the imperfect information game, where players bet on point cards with one of their 3 cards. The point cards are dealt in descending order (3,2,1) and players only know if they won, lost, or drew in the round(not the opponent's move). In the end, reward 1 is given to the player who won the most points. This game contains two pure equilibria (one of them is playing cards 3, 2, 1 in each round).
Blotto is a game parametrized by the number of fields and coins (sometimes called "troops"). Players redistribute their coins along the fields in order to have more coins in more fields than their opponent. With 3 fields and at least 6 coins, the game is not trivial and does not have a pure equilibrium. 
Kuhn Poker, a toy variant of Poker, in the original extensive form has sequential moves and a chance player responsible for distributing the cards.

\subsection{Exploitability and $\epsilon$-dominance}
\label{subsec:epsilon_exploitability}
As \epsdompure is built to minimize the upper bound on the exploitability of the portfolio, it was unclear how tight the two values are. The result of the experiment testing this connection is on Figure~\ref{fig:exploitability_vs_epsilon}. Not only are the two values correlated, but the exploitability is more than twice as low. That makes this method good for portfolio search. Additionally, in the Figure~\ref{subfig:size_vs_epsilon} we show that the size of the portfolio with respect to $\epsilon$ is high for small values but relatively quickly drops down, which is useful as a portfolio of a smaller size is more desirable.
In this experiment, we modified \epsdompure to instead minimize the size of the portfolio for a given constant value of $\epsilon$. Then such a program was run on $100$ random games of sizes $5,10$ and $15$.

\begin{figure}[h]
     \centering
     \begin{subfigure}[b]{0.234\textwidth}
         \centering
         \includegraphics[width=\textwidth]{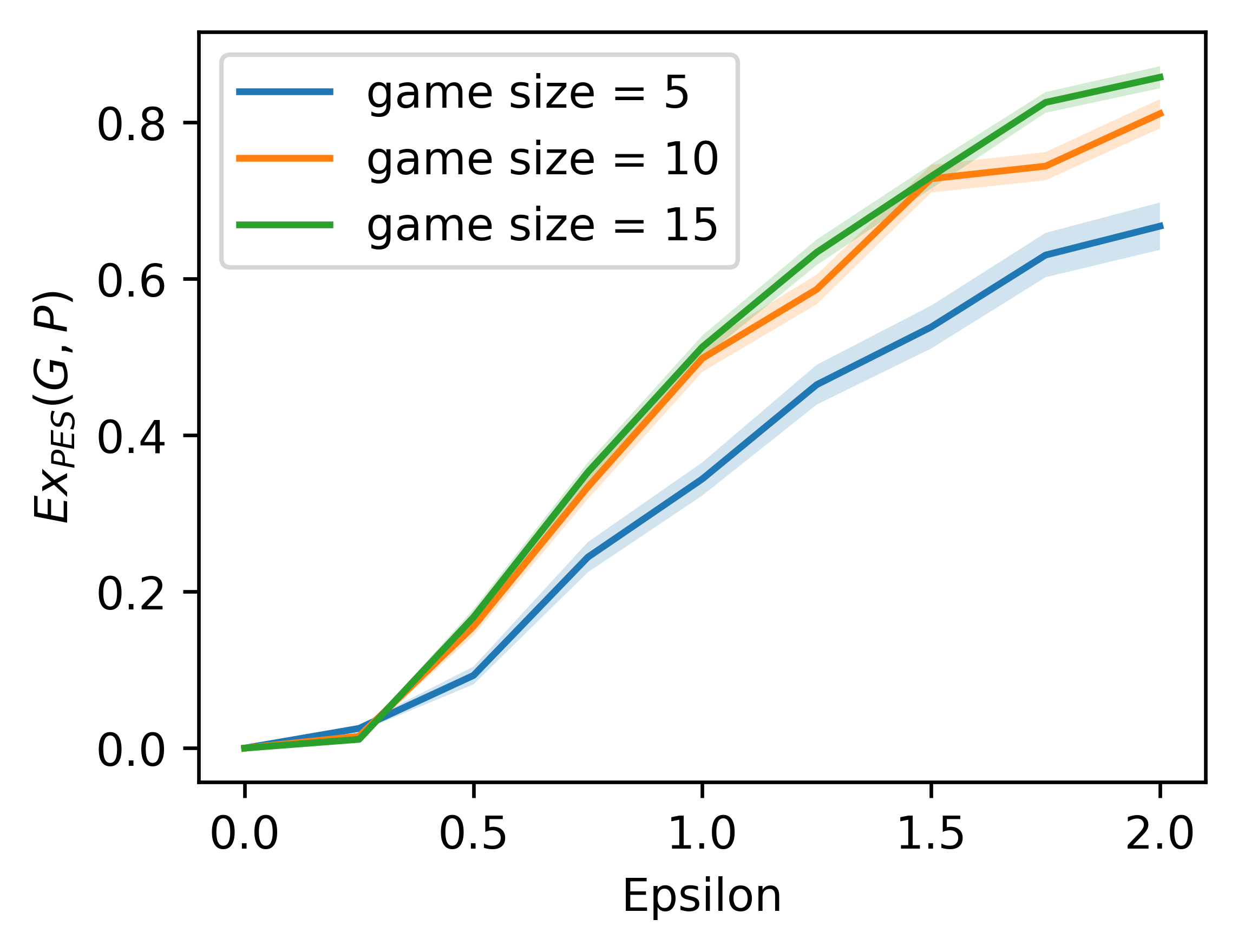}
        \caption{Pessimistic exploitability of portfolio with bounded epsilon}
        \label{subfig:exploitability_vs_epsilon}
     \end{subfigure}
     \begin{subfigure}[b]{0.234\textwidth}
         \centering
         \includegraphics[width=\textwidth]{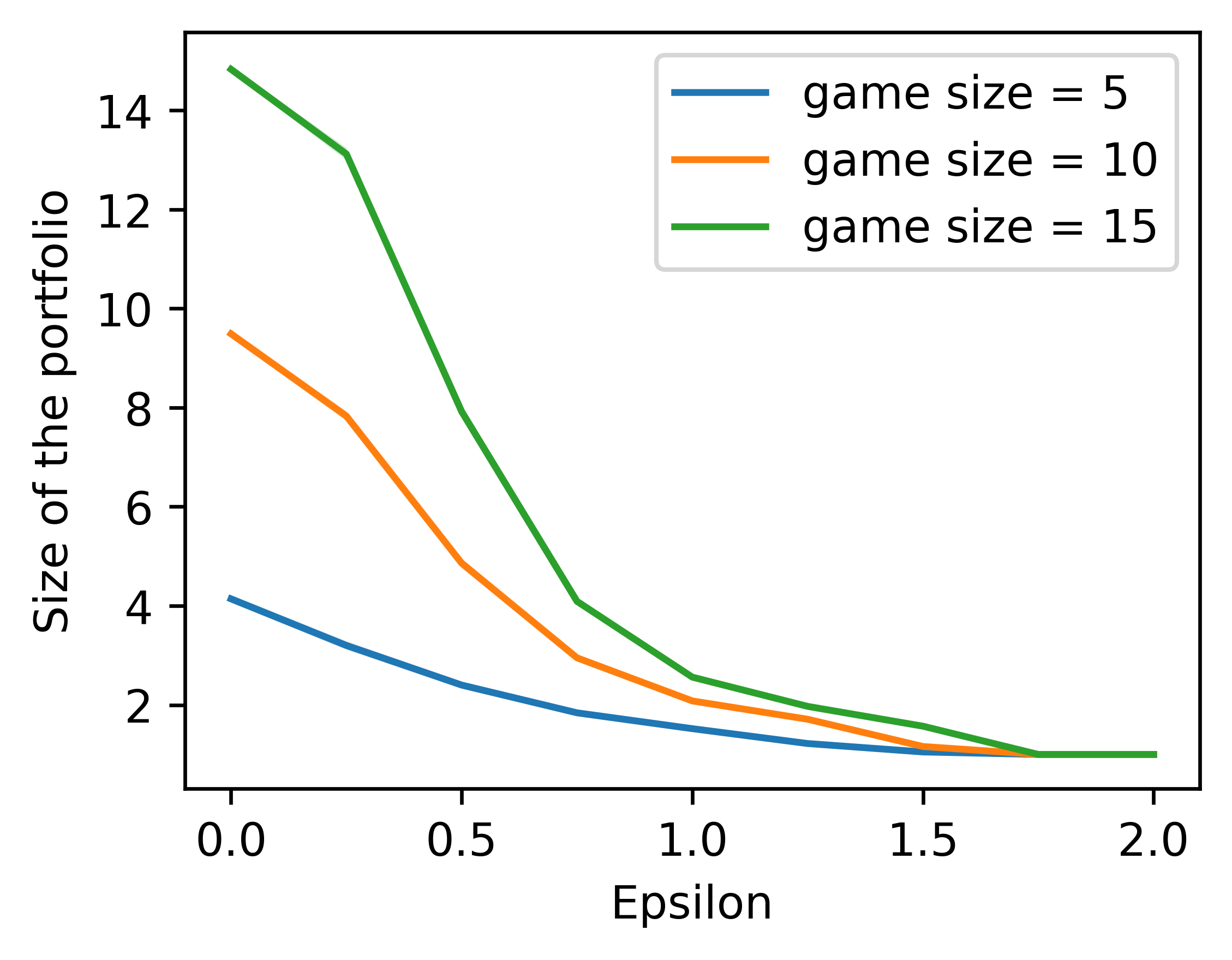}
        \caption{Size of the portfolio that is $\epsilon$-dominating set}
        \label{subfig:size_vs_epsilon}
     \end{subfigure}
     \caption{Experiment showing relation of pessimistic exploitability $ex_{PES}$ and size of the portfolio with different $\epsilon$ values. Portfolios were found by \epsdompure, with a changed objective to minimize its portfolio size for the given $\epsilon$ bound. Performed on $N=100$ random games, the mean taken and the standard error are shown by the shaded area. }
    \label{fig:exploitability_vs_epsilon}
\end{figure}

\subsection{Evaluation of \epsdompure and \epsdommixed }
\label{subsec:eps_dom_vs_greedyk}

\begin{figure}[h]
     \centering
     \includegraphics[width=0.24\textwidth]{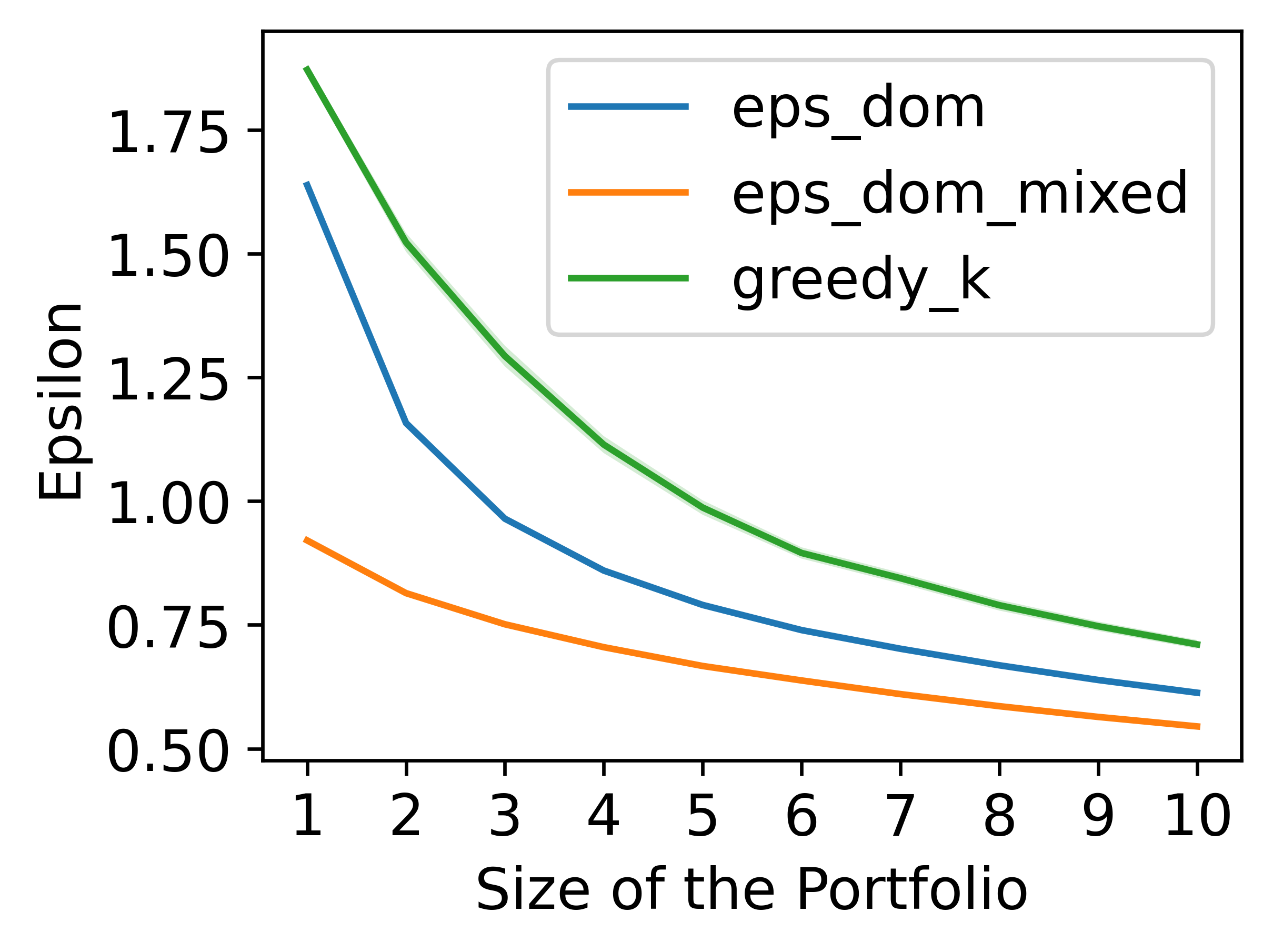}
 \caption{ Comparison of \epsdompure, \epsdommixed and Greedy-K on games of action sizes $|A_1|=|A_2|=25$. The experiment was run on $100$ different random games, the mean and the standard error are shown by the shaded area.}
\label{fig:eps_dom_vs_greedy_k}

\end{figure}

We evaluate the effectiveness of \epsdompure and \epsdommixed in identifying a subset of actions that are $\epsilon$-dominated with a minimal $\epsilon$ and compare them against the Greedy-K algorithm proposed by~\citet{Cheng_Weaker_than_weak_dominance_2007}. It's important to note that the Greedy-K algorithm was originally designed for computing iterated $\epsilon$-dominated Equilibria, a distinct objective from our portfolio selection problem. Given a budget for all eliminations, in each iteration set of $K$ strategies is removed. The set is chosen greedily by selecting actions which individually have the lowest $\epsilon$, then the joint $\epsilon$ for the set is computed with a linear program. This part is the part we will compare to, restricting Greedy-K to only remove $|K|=|A_2|-k$ strategies of player 2, disregarding the budget. We then compare the $\epsilon$ of the set $K$ with the $\epsilon$ of the sets $ A_2\setminus P$ that \epsdompure and \epsdommixed return. This modification allows for a direct comparison between the performance of a greedy selection based on $\epsilon$ values and the simultaneous selection achieved by our MILP-based \epsdompure approach.

\begin{figure*}[hbt!]
\includegraphics[width=15cm]{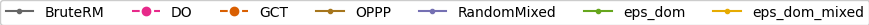}
     \centering
          \begin{subfigure}[b]{0.3\textwidth}         \includegraphics[width=5cm]{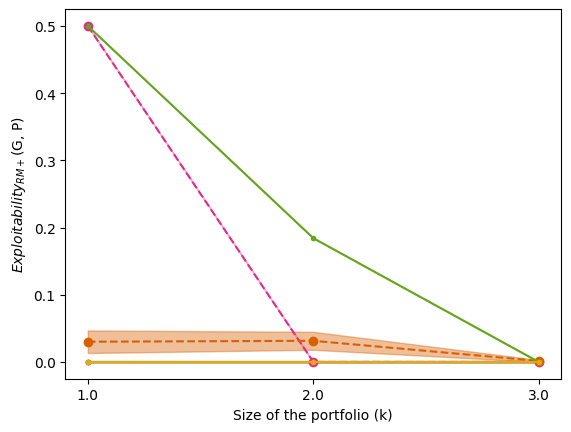}
         \caption{Goofspiel with 3 cards}
         \label{fig:efg_goofspiel}
    \end{subfigure}
     \begin{subfigure}[b]{0.3\textwidth}
         \includegraphics[width=5cm]{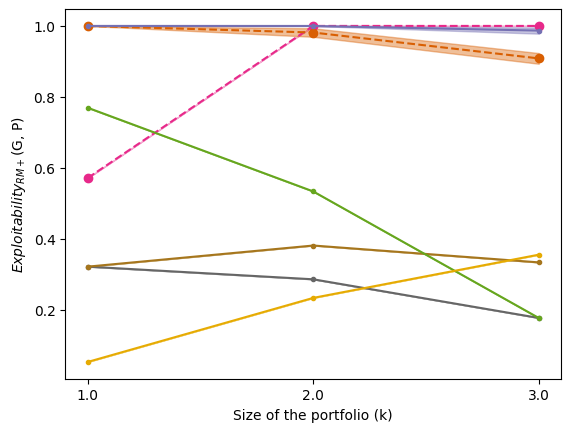}
         \caption{Blotto with 3 fields  and 8 coins}
         \label{fig:efg_blotto}
    \end{subfigure}
    \begin{subfigure}[b]{0.3\textwidth}
         \includegraphics[width=5cm]{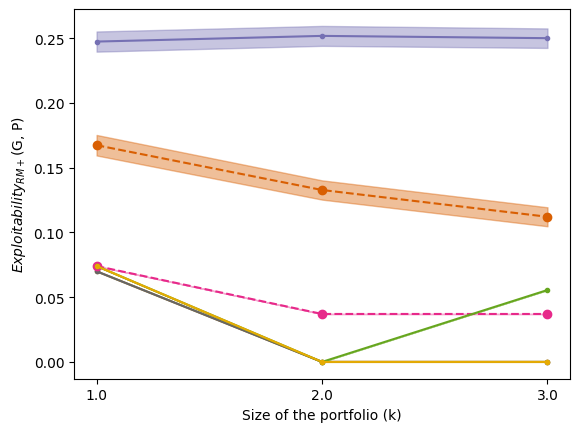}
         \caption{Kuhn Poker}
         \label{fig:efg_kuhn}
    \end{subfigure}
    \caption{RM+ Exploitability of different methods on extensive form games. Online algorithms are marked with a dashed line. For methods that depend on stochasticity, results for 50 different seeds were computed. The average and standard error is shown. GCT method uses a scaling factor of 0.3 and gradients from 10k iterations. }
    \label{fig:all_methods_efg}
\end{figure*}

In the Figure~\ref{fig:eps_dom_vs_greedy_k}, it is visible that global optimizing for a minimal $\epsilon$-dominating set gives significantly better results, as our algorithms outperform Greedy-K in this setup. Additionally, we see that the gap between the algorithms narrows while portfolio size increases, as expected, since the number of selected strategies to be removed is smaller. The difference in performance between \epsdompure and \epsdommixed further highlights the potential gains from considering mixed strategies.

While the plots for other game sizes (15 and 20) are included in Appendix~\ref{app:eval_eps_dom_greedy}, the observed trends remain consistent across all tested dimensions.

\subsection{Evaluation on Real Games}
\label{subsec:results_efg}

Finally, we evaluate our methods against a suite of established heuristics on standard benchmark games. We compare our methods,  \epsdompure and \epsdommixed, with the following baselines:

\begin{itemize}
    \item Optimal Pure Pessimistic Portfolio (OPPP) and Optimal Pure RM+ Portfolio (BruteRM): The ground-truth best pure portfolio for given exploitability, found via brute-force search.
    \item Double Oracle (DO)~\citep{mcmahan2003planning} method. Each player's strategy subset is iteratively expanded by adding best responses to the opponent's current actions. We initialize it with responses to the uniform strategy and continue until player 2's subset reaches the target portfolio size.
    \item Gradient Cluster Transformations (GCT)~\citep{kubicek-sepot-2024}, this online, scalable method constructs a portfolio in three steps. First, it calculates gradients from an algorithm converging towards Nash Equilibrium (though not necessarily reaching it). Second, these gradients are normalized and grouped into $k$ clusters, identifying key directions in the strategy space. Finally, a gradient is sampled from each cluster's center, scaled by a factor c, and combined with a base strategy (e.g., an approximate NE) to generate the portfolio. We adopt GCT to our 2p0s games setting and, following the original work, employ Regularized Nash Dynamics (R-NaD)~\citep{perolat2022mastering,perolat_poincare_2020} to compute the gradients.
    \item Random Mixed: A baseline that generates random mixed-strategy portfolios. Due to drawing in the simplex, the strategies tend to be close to uniform strategies.
\end{itemize}

Figure~\ref{fig:all_methods_efg} shows the performance of all methods across the three games. Our \epsdommixed method is exceptionally strong, consistently performing as well as or better than the optimal pure portfolio benchmark (OPPP). This is a significant result, as our method is more scalable than brute-force search.

The game-specific results are also insightful. In Goofspiel(Subfigure \ref{fig:efg_goofspiel}), for instance, the optimal portfolio consists of strategies that are not part of a Nash Equilibrium. This is exactly the case when DO fails due to its reliance on NE support. Moreover, it poses a challenge for \epsdompure, which relies on dominance relationships where multiple actions have the same $\epsilon=1$. However, \epsdommixed overcomes this easily by mixing strategies. This highlights the robustness of our mixed-strategy approach. Across all games, our principled methods are highly competitive, demonstrating their value as a practical tool for portfolio construction.

Blotto (Subfigure \ref{fig:efg_blotto}) turned out to be a game in which the GCT and DO methods struggle. Since both of the methods rely on the NE support and its search, the cause of this is likely to be a big size of the support of NE. Additionally, \epsdommixed achieves very good results again.

In the Kuhn Poker (Subfigure \ref{fig:efg_kuhn}) methods finding a pure portfolio have low exploitability compared to GCT. \epsdommixed performs as good as OPPP and BruteRM.

Additionally, experimental setup and tables with extended results can be found in the Appendix~\ref{app:efg_more_results}. The results show similar trends among all the experiments. 

\section{Conclusions}
\label{sec:conclusions}

We began by formalizing the concept of portfolios and the optimal portfolio selection problem. We established its computational difficulty as NP-hard. Our theoretical analysis further revealed that intuitive heuristics are often unreliable, with common approaches like relying on Nash Equilibrium support or greedy construction leading to provably poor solutions. We introduced analytical tools, centered on the concept of $\epsilon$-dominance, that provide provable quality guarantees for a portfolio. Furthermore, we proposed \epsdompure and \epsdommixed programs for both constructing portfolios with bounded exploitability, enabling the discovery of approximate Nash Equilibria by the player. We defined an \evaluation program for pessimistic evaluation. Our empirical evaluation confirmed the value of these methods, demonstrating that they outperform common heuristics and establish a robust performance benchmark.
We show that our approximate algorithms are, in practice, close to the optimal solutions. Additionally, our results on real games (EFGs) suggest that the currently used methods leave room for improvement. 

The primary impact of this research is the establishment of the theoretical and practical groundwork necessary to advance portfolio construction. While our proposed offline algorithms are not directly scalable to the immense games where these techniques are most critical, they provide the essential tools for developing and benchmarking the next generation of heuristics. Future work should focus on creating scalable, online algorithms that are informed by these theoretical insights. Future directions include designing an online algorithm that would be suitable for large imperfect-information games or improving the GCT method. Additionally, it would be studying different portfolio use cases

\section*{Acknowledgements}
This research is supported by the Czech Science Foundation\\ (GA25-18353S) and the Grant Agency of the CTU in Prague \\(SGS23/184/OHK3/3T/13). Computational resources were provided by the e-INFRA CZ project (ID:90254), supported by the Ministry of Education, Youth and Sports of the Czech Republic.

\bibliography{aaai2026}
\clearpage
\newpage

\appendix

\section{Proof of Theorem~\ref{thm:arbitrary_ne} and Observation~\ref{obs:size_portfolio_suppport} }
\label{app:proof_arbitrary_ne}
Proof of Theorem~\ref{thm:arbitrary_ne}.

\begin{proof}
We will demonstrate it on the following matrix game:
$$
U = \begin{bmatrix}
1 & 0 & \frac{1}{2} \\
0 & 1 & \frac{1}{2} \\
0 & 0 & \frac{1}{2} \\
\end{bmatrix}
$$
The only Nash equilibrium strategy for the row player is $(\frac{1}{2},\frac{1}{2},0)$ and the equilibria for the column player are for any $p\in[0,1] : (\frac{p}{2},\frac{p}{2},1-p)$. In particular, one equilibrium of the column player is $(0,0,1)$. However, if we use only the third column as the portfolio, the row player can pick any strategy as a Nash equilibrium of the restricted game.

In the pessimistic case, she will pick $(1,0,0)$, which would be responded by $(0,1,0)$, causing the cost of $\frac{1}{2}$. In the maximum entropy case, player 1 would play uniformly over all actions, which can lead to arbitrary small payoff, if we add copies of the last row into the game.
    
\end{proof}

Proof of the Observation~\ref{obs:size_portfolio_suppport}.
\begin{proof}
    Proof is based on the game from the proof of Theorem~\ref{thm:unique_ne}, which is the game with the following utility matrix, for any small $\delta \in (0,\frac{1}{2})$:
    $$
U = \begin{bmatrix}
1 & \delta & \frac{1}{2} \\
\delta & 1 & \frac{1}{2} \\
0 & 0 & \frac{1}{2} \\
\end{bmatrix}
$$
    The NE of player 2 is $\pi_2= (0,0,1)$, therefore it's support is of size 1. The value of this game is $\frac{1}{2}$. There are three pure portfolios of size 1:
    \begin{itemize}
        \item $P=\{(1,0,0)\}$, here player 1  chooses strategy $(1,0,0)$, to which best response of player 2 is $(0,1,0)$, causing exploitability of $\frac{1}{2} - \delta >0$
        \item $P=\{(0,1,0)\}$, here player 1  chooses strategy $(0,1,0)$, to which best response of player 2 is $(1,0,0)$, causing exploitability of $\frac{1}{2} - \delta >0$
        \item $P=\{(0,0,1)\}$, here player 1  chooses strategy $(0,0,1)$, because we are considering pessimistic exploitability. To this best response of player 2 is $(0,1,0)$, causing exploitability of $\frac{1}{2} - 0 = \frac{1}{2} >0$
    \end{itemize}  
    Therefore, no pure portfolio of the size of the support of NE of player 2 has exploitability of 0.
\end{proof}

\section{Proof of Theorems \ref{thm:game_max_expl} and \ref{thm:game_pes_exploit}}
\label{app:bad_games_proof}
Below is the proof of Theorem~\ref{thm:game_max_expl}.

\begin{proof}

One such game has the payoff matrix $U= -1 \cdot I$ where $I$ is the identity matrix of size n:
$$
U = \begin{bmatrix}
-1 & 0 & \ldots & 0 \\
0 & -1 & \ldots & 0 \\
\vdots & \vdots & \ddots & \vdots \\
0 & 0 & \ldots & -1
\end{bmatrix}
$$

In this game any pure portfolio of size $k<n$ will not include at least one column $i$. In the restricted game, the row $i$ will contain only zeros and hence playing it with probability 1 will be optimal for the row player regardless of what the column player does. However, choosing such an NE by the row player will lead to exploitability 1 in the full game as the column player will best respond with playing action $i$.
    
\end{proof}

Below is the proof of Theorem~\ref{thm:game_pes_exploit}.

\begin{proof}

One such game has the payoff matrix $-1 \cdot I$ where $I$ is the identity matrix.
$$
U = \begin{bmatrix}
-1 & 0 & \ldots & 0 \\
0 & -1 & \ldots & 0 \\
\vdots & \vdots & \ddots & \vdots \\
0 & 0 & \ldots & -1
\end{bmatrix}
$$
The payoff range in this game is $\Delta = 1- 0 =1$

Assume we write the portfolio strategies into columns of a matrix P and denote the strategy player 2 plays over the portfolio as $\vec{p}$. Then the problem of computing the NE in the restricted game can be written as
$$
\max_{\vec{x}}\min_{\vec{p}} \vec{x} U P \vec{p}^T
$$
That means that the players are effectively playing a matrix game $G=U P \in \mathcal{R}^{n\times k}$

For games with $k$ columns, we know there is a NE such that the support of the player 1 strategy is at most $k$. Hence, at least one of these strategies is played with probability at least $\frac{1}{k}$. We denote this strategy as $i$. Now, if the column player plays the pure strategy $i$ in the full game, the reward for player 1 will be at most $-\frac{1}{k}$.

The equilibrium in the full game, where both players play uniformly over all strategies, gives player 1 an expected reward of $-\frac{1}{n}$. Therefore, the pessimistic exploitability of any portfolio is at least $-\frac{1}{n}- \left(-\frac{1}{k}\right)= \frac{n-k}{nk}= \frac{n-k}{nk} \Delta $
    
\end{proof}

\section{Proof of the Theorem~\ref{thm:incremental}}
\label{app:proof_incremental}
\begin{proof}
One such game is defined by the following utility matrix(where $b_0,...,b_3$ denote the column player actions):

\begin{gather*}
\begin{matrix}
&b_0 && b_1 && b_2 && b_3 \\
\end{matrix}\\
U=\begin{bmatrix}
-1 & 1 & -101 & -99 \\
1 & -0.8 & -99 & -101 \\
\end{bmatrix}
\end{gather*}

Here, the portfolio $P=\{b_0, b_1\}$ is an optimal portfolio of size 2 with pessimistic exploitability 0. However, expanding it by adding any of $b_2$ or $b_3$ actions to P causes player 1 strategy to be pure, therefore having the exploitability of $ex_{PES}(G,P) = -100 + 101 = 1$ 
\end{proof}

%%%%%%%%%%%%%%%%%%%%%%%%%%%%%%%%%%%%%%%%%%%%%%%%%%%%%
\section{Proof of \epsdommixed Upper Bound }
\label{app:proof_eps_dom_mixed_milp}

\begin{proof}

Let $G$ be a game $G=(A_1,A_2, u)$ and $P={\{s_1,s_2,...s_k\}}\subseteq \Delta(A_2)$ be a set of mixed strategies, such that $A_2$ is $\epsilon$-dominated in game $G(A_1, A_2 \cup P)$. A strategy $\pi \in NE(G(A_1, P))$ is  then $\pi \in \epsilon-NE(G(A_1,A_2\cup P)$.
Because any strategy from $P$ can be also played in $A_2$, (there exists a mapping from $\Delta(A_2 \cup P)$ to $\Delta(A_2)$, strategy $\pi$ is also $\epsilon$-NE of game $G$, $\pi \in \epsilon-NE(G)$ .

Since the \epsdommixed finds such subset $P\subseteq\Delta(A_2)$ that  $\forall_{s_2 \in\Delta(A_2)} \exists_{s'_2 \in P} \forall_{s_1 \in \Delta(A_1)} U(s_1, s'_2) \leq U(s_1, s_2) +\epsilon$ it fulfills the conditions for $P$, hence, $ex(P)\leq \epsilon$.
\end{proof}

Note, there might be mixed portfolios for which $\epsilon$ is smaller than \epsdommixed, as it does not consider linear combinations over strategies from mixed portfolios (that would be non-linear). Nevertheless, \epsdommixed finds a portfolio with $\epsilon$ at least as small as \epsdompure.

%%%%%%%%%%%%%%%%%%%%%%%%%%%%%%%%%%%%%%%%%%%%
\section{Evaluation of \epsdompure and \epsdommixed }
\label{app:eval_eps_dom_greedy}
Additional experiments were conducted for sizes 15 and 20 and are present on the Figure~\ref{fig:eps_dom_vs_greedy_more_results}. They were conducted for $k\in\{1,...,10\}$.

\begin{figure}[h]
     \centering
     \begin{subfigure}[b]{0.45\textwidth}
         \centering
         \includegraphics[width=\textwidth]{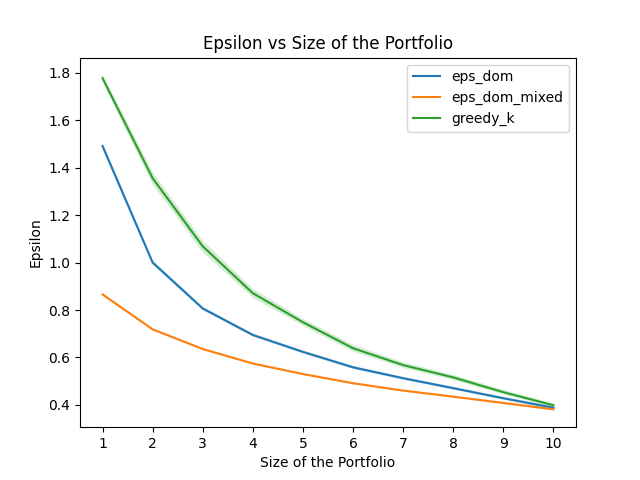}
        \caption{$|A_1|=|A_2|=15$}
     \end{subfigure}
     \begin{subfigure}[b]{0.45\textwidth}
         \centering
         \includegraphics[width=\textwidth]{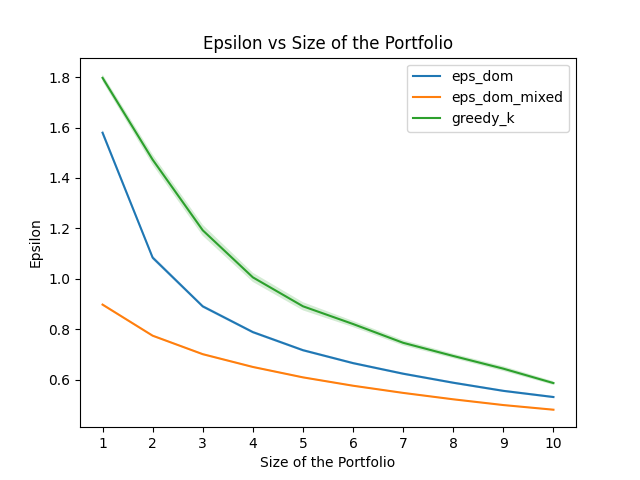}
        \caption{$|A_1|=|A_2|=20$}
     \end{subfigure}
     \caption{ Additional game sizes. Experiments on 100 random games. }
    \label{fig:eps_dom_vs_greedy_more_results}
\end{figure}

%%%%%%%%%%%%%%%%%%%%%%%%%%%%%%%%%%%%%%%%%%%%
\section{Proof of the  Observation~\ref{obs:mixed_pure} and Proposition~\ref{prp:mixed_vs_pure}}
\label{app:proof_mixed_vs_pure}
Observation~\ref{obs:mixed_pure} is trivially proven because mixed portfolios contain pure portfolios. Below we present the proof of the Proposition~\ref{prp:mixed_vs_pure}.
\begin{proof}

\begin{table}[h]
\centering
\begin{tabular}{|l|l|l|l|}
\hline
& \textbf{R} & \textbf{P} & \textbf{S} \\ \hline
\textbf{R} & 0 & -1 & 1 \\ \hline
\textbf{P} & 1 & 0 & -1 \\ \hline
\textbf{S} & -1 & 1 & 0 \\ \hline
\end{tabular}
  \caption{Utility matrix of Rock Paper Scissors(RPS)}
  \label{tab:rps_utility}
\end{table}

\begin{table}[h]
\centering
\begin{tabular}{|l|l|l|l|}
        \hline
         & \textbf{R} & \textbf{P} \\ \hline
        \textbf{R} & 0 & -1  \\ \hline
        \textbf{P} & 1 & 0  \\ \hline
        \textbf{S} & -1 & 1  \\ \hline
\end{tabular}
\caption{Utility matrix of $G(\{R,P,S\}, \allowbreak \{R, P\})$}
\label{tab:pure_ut}
\end{table}

\begin{table}[h]
\centering
\begin{tabular}{|l|l|l|l|}
        \hline
         & \textbf{$\sigma_1$} & \textbf{$\sigma_2$} \\ \hline
        \textbf{R} & -0.5 & 0 \\ \hline
        \textbf{P} & 0.5 & -0.5 \\ \hline
        \textbf{S} & 0 & 0.5 \\ \hline
\end{tabular}
\caption{Utility matrix of $G(\{R,P,S\}, \allowbreak \{(0.5, 0.5, 0.0), \allowbreak (0.0, 0.5, 0.5)\})$}
\label{tab:mixed_ut}
\end{table}

We prove the opposite by counterexample. We will show that for a specific game, there exists a mixed portfolio with a lower exploitability than any pure portfolio. We will show that for optimistic and pessimistic utility, which is enough because those types bound others.
Let $G$ be a Rock Paper Scissors game (Table~\ref{tab:rps_utility}) and portfolio size be $k=2$. In this case, an optimal portfolio is any portfolio that contains two distinct actions, then without loss of generality $P_{pure} = \{ (1,0,0), (0,1,0) \}$ is the optimal pure portfolio. Utility of restricted game $G(\{R,P,S\},P_{pure})$ is presented in the Table~\ref{tab:pure_ut}.
The Ne found by optimistic and pessimistic NE selecting function both result in the following strategy of player 1: $\pi^*_{1} = (0, \frac{2}{3}, \frac{1}{3})$ and best response of player 2 in the original game $G$ to be $BR_2(\pi^*_{1}) =(0,0,1)$. Therefore, the exploitability of this portfolio is: $ex_{OPT}(G, P_{pure}) = ex_{PES}(G, P_{pure}) = |0 - \frac{2}{3}| \approx 0.67$

Now, we will show a mixed portfolio of size $k=2$, that has a lower exploitability for this game than $0.67$ with $P_{mixed} = \{(0.5, 0.
5, 0) ,  (0., 0.5, 0.5) \}$. Utility of restricted game $G(\{R,P,S\},P_{mixed})$ is presented in the Table~\ref{tab:mixed_ut}.
Again, the optimistic and pessimistic NE both result in the following strategy of player 1: $\pi^*_{1}  =(0, \frac{1}{3}, \frac{2}{3})$, with best response of player 2  to be $BR_2(\pi^*_{1}) =(0,0,1)$.  Therefore the exploitability of this portfolio is: $ex_{OPT}(G, P_{mixed}) =ex_{PES}(G, P_{mixed}) = | 0 - \frac{1}{3}| \approx 0.33$

\end{proof}

%%%%%%%%%%%%%%%%%%%%%%%%%%%%%%%%%%%%%%%%%%%%

\section{Experimental Setup and More Results on EFG Games}
\label{app:efg_more_results}

\begin{table*}[t]
\centering
\begin{tabular}{l | ccc | ccc}
\hline
$ex_{RM+}$ & \multicolumn{3}{c|}{Blotto (F=3,C=6)} & \multicolumn{3}{c}{Blotto (F=3,C=7)} \\
\hline
k & 1 & 2 & 3 & 1 & 2 & 3 \\
\hline
OPPP        & 0.27 & 0.5  & 0.33 & 0.29 & 0.47 & 0.50 \\
DO          & 0.33 & 0.67 & 0.88 & 0.4  & 0.4  & 1.0  \\
GCT         & 1 $\pm$ 0.00 & 0.96$\pm$0.01 & 0.86$\pm$0.01 & 1$\pm$0.00 & 0.98$\pm$0.01 & 0.92$\pm$0.01 \\
\epsdompure & 0.71 & 0.75 & 0.24 & 0.40 & 1 & 0.15 \\
\epsdommixed& 0.24 & 0.33 & 0.33 & 0.09 & 0.56 & 0.55 \\
RandomMixed & 1$\pm$0.00 & 0.96$\pm$0.01 & 0.94$\pm$0.01 & 1$\pm$0.00 & 0.99$\pm$0.00 & 0.99$\pm$0.01 \\
OPP RM+     & 0.27 & 0.3  & 0.24 & 0.29 & 0.27 & 0.15 \\
\hline
\end{tabular}
\caption{}
\label{tab:blotto67_rm}
\end{table*}

\begin{table*}[t]
\centering
\begin{tabular}{l | ccc}
\hline
$ex_{RM+}$ & \multicolumn{3}{c}{Blotto (F=3,C=8)} \\
\hline
k & 1 & 2 & 3 \\
\hline
OPPP        & 0.32 & 0.38 & 0.33 \\
DO          & 0.57 & 1.0  & 1.0  \\
GCT         & 1$\pm$0.00 & 0.97$\pm$0.01 & 0.90 $\pm$ 0.01 \\
\epsdompure & 0.57 & 1 & 0.18 \\
\epsdommixed& 0.08 & 0.23 & 0.35 \\
RandomMixed & 1$\pm$0.00 & 0.98$\pm$0.01 & 0.96$\pm$0.01 \\
OPP RM+     & 0.32 & 0.29 & 0.18 \\
\hline
\end{tabular}
\caption{}
\label{tab:blotto8_rm}
\end{table*}

\begin{table*}[t]
\centering
\begin{tabular}{l | ccc | ccc}
\hline
$ex_{PES}$ & \multicolumn{3}{c|}{Blotto (F=3,C=6)} & \multicolumn{3}{c}{Blotto (F=3,C=7)} \\
\hline
k & 1 & 2 & 3 & 1 & 2 & 3 \\
\hline
OPPP        & 1 & 1 & 0.33 & 1 & 1 & 0.5 \\
DO          & 1 & 1 & 1    & 1 & 1 & 1   \\
GCT         & 1$\pm$0.00 & 0.98$\pm$0.01 & 0.87$\pm$0.01 & 1$\pm$0.00 & 0.98$\pm$0.01 & 0.93$\pm$0.01 \\
\epsdompure & 1 & 1 & 1    & 1 & 1 & 1   \\
\epsdommixed& 1 & 0.67 & 0.33 & 1 & 1 & 1 \\
RandomMixed & 1$\pm$0.00 & 0.96$\pm$0.01 & 0.94$\pm$0.01 & 1$\pm$0.00 & 1$\pm$0.00 & 0.98$\pm$0.01 \\
OPP RM+     & 1 & 1 & 1 & 1 & 1 & 1 \\
\hline
\end{tabular}
\caption{}
\label{tab:blotto67_pessimistic}
\end{table*}

\begin{table*}[t]
\centering
\begin{tabular}{l | ccc}
\hline
$ex_{PES}$ & \multicolumn{3}{c}{Blotto (F=3,C=8)} \\
\hline
k & 1 & 2 & 3 \\
\hline
OPPP        & 1 & 1 & 0.33 \\
DO          & 1 & 1 & 1 \\
GCT         & 1$\pm$0.00 & 0.99$\pm$0.01 & 0.90$\pm$0.01 \\
\epsdompure & 1 & 1 & 1 \\
\epsdommixed& 1 & 1 & 1 \\
OPP RM+     & 1 & 1 & 1 \\
\hline
\end{tabular}
\caption{}
\label{tab:blotto8_pessimistic}
\end{table*}

\begin{table*}[t]
\centering
\begin{tabular}{l | ccc | ccc}
\hline
$ex_{RM+}$ & \multicolumn{3}{c|}{Kuhn Poker (b=1.5)} & \multicolumn{3}{c}{Kuhn Poker (b=2)} \\
\hline
k & 1 & 2 & 3 & 1 & 2 & 3 \\
\hline
OPPP        & 0.05 & 0.01 & -    & 0.07 & 0 & - \\
DO          & 0.05 & 0.04 & 0.04 & 0.07 & 0.04 & 0.04 \\
GCT         & 0.26$\pm$0.01 & 0.17$\pm$0.01 & 0.14$\pm$0.01 & 0.25$\pm$0.01 & 0.16$\pm$0.01 & 0.12$\pm$0.01 \\
\epsdompure & 0.05 & 0.03 & 0.04 & 0.07 & 0 & 0 \\
\epsdommixed& 0.05 & 0.03 & 0.03 & 0.07 & 0 & 0 \\
RandomMixed & 0.16$\pm$0.00 & 0.16$\pm$0.00 & 0.16$\pm$0.00 & 0.25$\pm$0.01 & 0.25$\pm$0.01 & 0.25$\pm$0.01 \\
OPP RM+     & 0.05 & 0.01 & - & 0.02 & 0 & - \\
\hline
\end{tabular}
\caption{}
\label{tab:kuhn15_2_rm}
\end{table*}

\begin{table*}[t]
\centering
\begin{tabular}{l | ccc | ccc}
\hline
$ex_{RM+}$ & \multicolumn{3}{c|}{Kuhn Poker (b=2.5)} & \multicolumn{3}{c}{Kuhn Poker (b=3)} \\
\hline
k & 1 & 2 & 3 & 1 & 2 & 3 \\
\hline
OPPP         & 0.12 & 0 & -    & 0.17 & 0 & - \\
DO           & 0.12 & 0.06 & 0.06 & 0.08 & 0.13 & 0.11 \\
GCT          & 0.12$\pm$0.02 & 0.07$\pm$0.01 & 0.06$\pm$0.01 & 0.25$\pm$0.02 & 0.08$\pm$0.01 & 0.05$\pm$0.01 \\
\epsdompure  & 0.10 & 0 & 0.04 & 0.08 & 0 & 0.07 \\
\epsdommixed & 0.03 & 0 & 0.06 & 0 & 0 & 0.04 \\
RandomMixed  & 0.34$\pm$0.01 & 0.34$\pm$0.01 & 0.34$\pm$0.01 & 0.31$\pm$0.02 & 0.27$\pm$0.02 & 0.24$\pm$0.02 \\
OPP RM+      & 0.05 & 0 & - & 0.08 & 0 & - \\
\hline
\end{tabular}
\caption{}
\label{tab:kuhn_25_3_rm}
\end{table*}

\begin{table*}[t]
\centering
\begin{tabular}{l | ccc | ccc}
\hline
$ex_{PES}$ & \multicolumn{3}{c|}{Kuhn Poker (b=1.5)} & \multicolumn{3}{c}{Kuhn Poker (b=2)} \\
\hline
k & 1 & 2 & 3 & 1 & 2 & 3 \\
\hline
OPPP         & 0.05 & 0.01 & - & 0.07 & 0 & - \\
DO           & 0.05 & 0.04 & 0.04 & 0.07 & 0.04 & 0.04 \\
GCT          & 0.26$\pm$0.01 & 0.17$\pm$0.01 & 0.14$\pm$0.01 & 0.26$\pm$0.01 & 0.16$\pm$0.01 & 0.12$\pm$0.01 \\
\epsdompure  & 0.05 & 0.03 & 0.04 & 0.07 & 0 & 0 \\
\epsdommixed & 0.05 & 0.03 & 0.03 & 0.07 & 0 & 0 \\
RandomMixed  & 0.16$\pm$0.00 & 0.16$\pm$0.00 & 0.16$\pm$0.00 & 0.25$\pm$0.01 & 0.25$\pm$0.01 & 0.25$\pm$0.01 \\
OPP RM+      & 0.05 & 0.01 & - & 0.19 & 0 & - \\
\hline
\end{tabular}
\caption{}
\label{kuhn15_2_pessimistic}
\end{table*}

\begin{table*}[t]
\centering
\begin{tabular}{l | ccc | ccc}
\hline
$ex_{PES}$ & \multicolumn{3}{c|}{Kuhn Poker (b=2.5)} & \multicolumn{3}{c}{Kuhn Poker (b=3)} \\
\hline
k & 1 & 2 & 3 & 1 & 2 & 3 \\
\hline
OPPP         & 0.12 & 0 & - & 0.17 & 0 & - \\
DO           & 0.12 & 0.06 & 0.06 & 0.17 & 0.33 & 0.28 \\
GCT          & 0.12$\pm$0.02 & 0.07$\pm$0.01 & 0.05$\pm$0.01 & 0.26$\pm$0.03 & 0.08$\pm$0.01 & 0.06$\pm$0.01 \\
\epsdompure  & 0 & 0 & 0.04 & 0 & 0 & 0.08 \\
\epsdommixed & 0.03 & 0 & 0.06 & 0 & 0 & 0.04 \\
RandomMixed  & 0.34$\pm$0.01 & 0.34$\pm$0.01 & 0.34$\pm$0.01 & 0.31$\pm$0.02 & 0.27$\pm$0.02 & 0.24$\pm$0.02 \\
OPP RM+      & 0.17 & 0 & - & 0.17 & 0 & - \\
\hline
\end{tabular}
\caption{}
\label{tab:kuhn_25_3_pessimistic}
\end{table*}

\begin{table*}[t]
\centering
\begin{tabular}{l | ccc}
\hline
$ex_{RM+}$ & \multicolumn{3}{c}{Oshi Zumo (coins=4, horizon=2, min\_bid=1, size=3)} \\
\hline
k & 1 & 2 & 3 \\
\hline
DO           & 0.49 & 0.49 & 0.37 \\
GCT          & 0.27$\pm$0.03 & 0.01$\pm$0.01 & 0.0$\pm$0.0 \\
\epsdompure  & 0.58 & 0.37 & 0 \\
\epsdommixed & 0 & 0 & 0 \\
\hline
\end{tabular}
\caption{}
\label{tab:oshi_rm}
\end{table*}

\begin{table*}[t]
\centering
\begin{tabular}{l | ccc}
\hline
$ex_{PES}$ & \multicolumn{3}{c}{Oshi Zumo (coins=4, horizon=2, min\_bid=1, size=3)} \\
\hline
k & 1 & 2 & 3 \\
\hline
DO           & 1 & 1 & 1 \\
GCT          & 0.62$\pm$0.07 & 0.02$\pm$0.02 & 0.01$\pm$0.01 \\
\epsdompure  & 1 & 1 & 0 \\
\epsdommixed & 0 & 0 & 0 \\
RandomMixed  & 0$\pm$0.0 & 0$\pm$0.0 & 0$\pm$0.0 \\
\hline
\end{tabular}
\caption{}
\label{tab:oshi_zumo_pessimistic}
\end{table*}

\begin{table*}[t]
\centering
\begin{tabular}{l | ccc}
\hline
$ex_{RM+}$ & \multicolumn{3}{c}{Goofspiel 3} \\
\hline
k & 1 & 2 & 3 \\
\hline
OPPP        & 0 & 0 & 0 \\
DO          & 0.5 & 0 & 0 \\
GCT         & 0.04$\pm$0.01 & 0.04$\pm$0.01 & 0$\pm$0.00 \\
\epsdompure & 0.13 & 0.06 & 0.08 \\
\epsdommixed& 0 & 0 & 0 \\
RandomMixed & 0$\pm$0.00 & 0$\pm$0.00 & 0$\pm$0.00 \\
OPP RM+     & 0 & 0 & 0 \\
\hline
\end{tabular}
\caption{}
\label{tab:goofspiel_rm}
\end{table*}

\begin{table*}[t]
\centering
\begin{tabular}{l | ccc}
\hline
$ex_{PES}$ & \multicolumn{3}{c}{Goofspiel 3} \\
\hline
k & 1 & 2 & 3 \\
\hline
OPPP        & 0 & 0 & 0 \\
DO          & 1 & 0 & - \\
GCT         & 0.06$\pm$0.02 & 0.07$\pm$0.02 & 0.01$\pm$0.01 \\
\epsdompure & 1 & 0.5 & 0 \\
\epsdommixed& 0 & 0 & 0 \\
RandomMixed & 0$\pm$0.00 & 0$\pm$0.00 & 0$\pm$0.00 \\
OPP RM+     & 0 & 0 & 0 \\
\hline
\end{tabular}
\caption{}
\label{tab:goofspiel_pessimistic}
\end{table*}

\subsection{Experimental Setup}
For the RM+ evaluation, the number of steps for RM+ algorithm was 10k. GCT algorithm was also run with 10k RNAD steps, which results in up to 10k gradients. 

The experiments were run for seeds 10,11,...,59. For each run, the seed was used to set a numpy random generator, which was then used to generate any pseudo-random numbers, e.g. utility for a random game.

Experiments were run on the experimental cluster, only using CPU resources. Each experiment had a reserved 32 CPU cores and 10 GB of memory.

\subsection{More Results}
Blotto was run with 3 Fields and with 3 different coins values(6,7,8). This changes the size of the game. Results from pessimistic exploitability on Blotto are presented in the Table~\ref{tab:blotto67_pessimistic} and \ref{tab:blotto8_pessimistic}. In the Table~\ref{tab:blotto67_rm} and \ref{tab:blotto8_rm} RM+ exploitability results are shown.
\subsubsection{Kuhn Poker}
Kuhn poker was tested for different values of bet value (1.5, 2, 2.5, 3). Results from pessimistic exploitability on Kuhn Poker are presented in the Table~\ref{kuhn15_2_pessimistic} and~\ref{tab:kuhn_25_3_pessimistic}. In the Table~\ref{tab:kuhn15_2_rm} and \ref{tab:kuhn_25_3_rm} RM+ exploitability results are shown. 
In this game, for portfolio size $k=3$, results for methods OPPP and BruteRM are not presented as the methods required more than 15 hours to compute. 

\subsubsection{Goofspiel}
Results from pessimistic exploitability on Goofspiel 3 are presented in the Table~\ref{tab:goofspiel_pessimistic}  and in the Table~\ref{tab:goofspiel_rm} RM+ exploitability.

\subsubsection{Oshi Zumo}
Oshi Zumo is the additional domain of EFG that we have tested the methods on. However, we saw a similar tendency as in Goofspiel 3 game. Therefore, we left it for appendix. 

Oshi Zumo was tested on a variant with 4 coins(budget of the player, horizon 2 (how many moves can be played), minimum bid of 1, and size of 3 (for one player, the whole board is of size 7). Results with pessimistic exploitability on Oshi Zumo are presented in the Table~\ref{tab:oshi_zumo_pessimistic} and in the Table~\ref{tab:oshi_rm} RM+ exploitability.
\end{document}